\documentclass[10pt,conference]{IEEEtran}

\usepackage{theorem}
\usepackage{enumerate}
\usepackage{amsmath}
\usepackage{graphicx} 
\usepackage{stmaryrd}
\usepackage{subfigure}
\usepackage[normalem]{ulem}
\usepackage{mathdots}
\usepackage{bbm}
\usepackage{psfrag}
\usepackage{url}

\IEEEoverridecommandlockouts

{\theoremstyle{plain}
   \newtheorem{theorem}{Theorem}}

{\theoremstyle{plain}
   \newtheorem{lemma}{Lemma}}

\newcommand{\concc}{\ensuremath{\mathcal{C}}}
\newcommand{\innerc}{\ensuremath{\mathcal{C}^{\mathrm{i}}}}
\newcommand{\outerc}{\ensuremath{\mathcal{C}^{\mathrm{o}}}}
\newcommand{\concn}{\ensuremath{n}}
\newcommand{\innern}{\ensuremath{n^{\mathrm{i}}}}
\newcommand{\outern}{\ensuremath{n^{\mathrm{o}}}}
\newcommand{\conck}{\ensuremath{k}}
\newcommand{\innerk}{\ensuremath{k^{\mathrm{i}}}}
\newcommand{\outerk}{\ensuremath{k^{\mathrm{o}}}}
\newcommand{\concd}{\ensuremath{d}}
\newcommand{\innerd}{\ensuremath{d^{\mathrm{i}}}}
\newcommand{\outerd}{\ensuremath{d^{\mathrm{o}}}}
\newcommand{\concerrors}{\ensuremath{e}}
\newcommand{\concerrorsfail}{\ensuremath{e_\mathrm{fail}}}
\newcommand{\outererrors}{\ensuremath{\varepsilon(k)}}
\newcommand{\outererrorsexplicit}[1]{\ensuremath{\varepsilon(#1)}}
\newcommand{\outererasures}{\ensuremath{\tau(k)}}
\newcommand{\outererasuresexplicit}[1]{\ensuremath{\tau(#1)}}
\newcommand{\innercw}[1]{\ensuremath{\mathbf{#1}^{\mathrm{i}}}}
\newcommand{\innertcw}[1]{\ensuremath{\mathbf{#1}^{\mathrm{i}, T}}}
\newcommand{\outercw}[1]{\ensuremath{\mathbf{#1}^{\mathrm{o}}}}
\newcommand{\Hd}[1]{\ensuremath{\mathrm{d}_\mathrm{H}(#1)}}
\newcommand{\erasure}{\ensuremath{\vartimes}}
\newcommand{\floor}[1]{\ensuremath{\lfloor #1\rfloor}}
\newcommand{\biggfloor}[1]{\ensuremath{\bigg\lfloor #1\bigg\rfloor}}
\newcommand{\bigceiling}[1]{\ensuremath{\big\lceil #1\big\rceil}}
\newcommand{\F}{\ensuremath{\mathbbm{F}}}
\newcommand{\mat}[1]{\ensuremath{#1}}

\psfrag{z, z}{$\underline{z}$, $z^\star$}
\psfrag{minimum inner distance di spspspspspsp}{minimum inner distance $\innerd$}
\psfrag{uz, l = 2 sp}{$\scriptscriptstyle \underline{z},\;\ell=2$}
\psfrag{zs, l = 2 sp}{$\scriptscriptstyle z^\star,\;\ell=2$}
\psfrag{uz, l = 8 sp}{$\scriptscriptstyle \underline{z},\;\ell=8$}
\psfrag{zs, l = 8 sp}{$\scriptscriptstyle z^\star,\;\ell=8$}
\psfrag{ell = 2}{$\scriptscriptstyle \ell=2$}
\psfrag{decoding bound sp}{decoding bound}
\psfrag{number of thresholds z space}{number of thresholds $z$}
\psfrag{dodi/2 sp}{$\scriptscriptstyle \outerd\innerd/2$}

\begin{document}

\title{Decoding Generalized Concatenated Codes Using Interleaved Reed--Solomon Codes}

\author{
\authorblockN{Christian Senger, Vladimir Sidorenko, Martin Bossert}\thanks{This work has been supported by DFG,
Germany, under grants BO~867/15, and BO~867/17. Vladimir Sidorenko is on leave from IITP, Russian Academy of Sciences, Moscow, Russia.}
\authorblockA{\small Inst. of Telecommunications and Applied Information Theory\\
Ulm University, Ulm, Germany \\
\{christian.senger$\;\vert\;$vladimir.sidorenko$\;\vert\;$martin.bossert\}@uni-ulm.de}
\and
\authorblockN{Victor Zyablov}
\authorblockA{\small Inst. for Information Transmission Problems\\
Russian Academy of Sciences, Moscow, Russia \\
zyablov@iitp.ru}
}

\maketitle

\begin{abstract}
Generalized Concatenated codes are a code construction consisting of a number of outer codes whose code symbols are protected by an inner code. As outer codes, we assume the most frequently used Reed--Solomon codes; as inner code, we assume some linear block code which can be decoded up to half its minimum distance. Decoding up to half the minimum distance of Generalized Concatenated codes is classically achieved by the Blokh--Zyablov--Dumer algorithm, which iteratively decodes by first using the inner decoder to get an estimate of the outer code words and then using an outer error/erasure decoder with a varying number of erasures determined by a set of pre-calculated thresholds. In this paper, a modified version of the Blokh--Zyablov--Dumer algorithm is proposed, which exploits the fact that a number of outer Reed--Solomon codes with average minimum distance $\bar{d}$ can be grouped into one single Interleaved Reed--Solomon code which can be decoded beyond $\bar{d}/2$. This allows to skip a number of decoding iterations on the one hand and to reduce the complexity of each decoding iteration significantly -- while maintaining the decoding performance -- on the other.
\end{abstract}

\section{Introduction}

In 1966, Forney introduced the concept of concatenated codes \cite{forney:1966}. It was generalized in 1976 by Blokh and Zyablov to Generalized Concatenated (GC) codes \cite{blokh_zyablov:1976}. The GC approach allows to design powerful codes with large block lengths using short, and thus easily decodable, component codes. The designed distance  and the performance of GC codes can be easily estimated theoretically. This allows to design GC codes for applications like e.g. optical lines, where block error rates in the order of $10^{-15}$ are required, a region, where simulations are not feasible. GC codes can be decoded up to half their mini\-mum distance using a sufficiently large number of decoding attempts for each outer code in the Blokh--Zyablov--Dumer algorithm (BZDA) \cite{blokh_zyablov:1982}, see also Dumer \cite{dumer:1980, dumer:1998}. We should also mention papers by Sorger \cite{sorger:1993}, and K\"{o}tter \cite{koetter:1993} who suggested interesting modifications of a BMD decoder in such a way that multi-attempt decoding of the outer code can be made ''in one step''.
Nielsen suggested in \cite{Nielsen:2001phd} to use the Guruswami--Sudan list decoding algorithm \cite{guruswami_sudan:1999} for decoding the outer codes and has shown that in this case also one decoding attempt is sufficient to allow decoding up to half the minimum distance of the GC code. In this paper, we employ another idea, which allows to decrease the number of outer decodings, but also to skip many decodings of the inner code. This idea is based on Interleaved Reed--Solomon (IRS) codes \cite{schmidt_sidorenko_bossert:2005b}. Other aspects of using IRS codes in concatenated codes were considered in \cite{justesen_thommesen_hoeholdt:2004, schmidt_sidorenko_bossert:2006c}.

The rest of our paper is organized as follows: In Section \ref{section:GCC} we explain GC codes as well as the required notations and assumptions. Section \ref{section:IRS} explains IRS codes and where they appear within GC codes. Section \ref{section:BZouterRS} gives an overview of the BZDA as introduced in \cite{blokh_zyablov:1982}. A generalization of this is given in Section \ref{section:BZouterIRS}, leading to a new algorithm which maintains the error-correcting performance of GC codes while reducing the number of outer decodings and skipping many inner decodings. In Section \ref{section:Example} we illustrate our results by means of some examples.

\section{Generalized Concatenated Codes and their Decoding}\label{section:GCC}

Encoding of a GC code $\concc(\concn, \conck, \concd)$ of order $\ell$ is as follows, where we restrict ourselves to outer RS codes $\outerc_l(\outern, \outerk_l, \outerd_l)$, $l=0,\ldots, \ell$, over the binary extension field $\F_{2^m}$ and an inner binary block code $\innerc(\innern, \innerk, \innerd)$ with dimension $\innerk=\ell m$.

The first step is outer encoding. For this we take $\ell$ codewords $\outercw{c}_l$ of the outer codes $\outerc_l$ and put them as rows into an $\ell\times \outern$ matrix 
\begin{equation*}
\mat{A}:=\left(
\begin{array}{c}
  \outercw{c}_0 \\
  \vdots \\
  \outercw{c}_{\ell-1}
\end{array}
\right),\;\outercw{c}_l\in \outerc_l
\end{equation*}
over $\F_{2^m}$. The second step is inner encoding, where the binary counterparts of the columns $j=0, \ldots, \outern-1$ of $\mat{A}$ are encoded by the inner code to obtain code word columns $\innertcw{c}_j\in\innerc$. The result of this procedure is a binary $\innern\times\outern$ matrix $\mat{C}:=
(\innertcw{c}_0, \ldots, \innertcw{c}_{\outern-1})$, which in turn is a code word of the GC code $\concc$.

The inner code $\innerc$ has the following nested structure: The code words are obtained by encoding an arbitrary binary information vector $(u_0, \ldots, u_{\innerk-1})$. If we fix $l\in \{0,\ldots, \ell-1\}$ groups of $m$ information bits starting from $u_0$ and encode, we obtain a subcode $\innerc_{l}(u_0, \ldots, u_{(l-1)m-1})\subseteq \innerc$ with distance $\innerd_{l}$. Note that $\innerd_{0}\leq \cdots \leq \innerd_{\ell-1}$ due to the special choice of encoders for the subcodes $\innerc_l$. Obviously, $\innerc_0=\innerc$. The minimum distance of the GC code $\concc$ is then lower bounded by its designed distance
\begin{equation}\label{eqn:designeddistance}
  \concd:=\mathrm{min}\{\outerd_0\innerd_0,\ldots, \outerd_{\ell-1}\innerd_{\ell-1}\}.
\end{equation}
The channel adds a binary error matrix $\mat{E}$ of weight $e$ to the transmitted concatenated code matrix $\mat{C}$, resulting in a received matrix $\mat{R}:=(\innertcw{r}_0, \ldots, \innertcw{r}_{\outern-1})=\mat{C}+\mat{E}$ at the receiver.

As output of GC decoding with the BZDA we want to obtain a matrix
\begin{equation*}
\mat{\bar{A}}:=(\bar{a}_{i, j})_{i=0,\ldots, \ell-1,\, j=0, \ldots, \outern-1}=
\left(
\begin{array}{c}
  \outercw{\bar{c}}_0 \\
  \vdots \\
  \outercw{\bar{c}}_{\ell-1}
\end{array}
\right),\;\outercw{\bar{c}}_l\in \outerc_l
\end{equation*}
over $\F_{2^m}$, which is an estimate of the matrix $\mat{A}$. Decoding consists of $\ell$ iterations, where iteration $l= 0,\ldots, \ell-1$ is as follows:
\begin{enumerate}[(i)]
  \item Decoding of all columns $\innertcw{r}_j$, $j= 0,\ldots, {\outern-1}$, of the received matrix $\mat{R}$ by BMD decoders for the nested subcodes $\innerc_{l}(\bar{a}_{0, j}, \ldots, \bar{a}_{l-1, j})$, correcting up to $\floor{(\innerd_{l}-1)/2}$ errors and altogether yielding a new estimate ${\outercw{\tilde{c}}}_{l}$ for the $l$-th row of $\mat{\bar{A}}$.
  \item Execution of $z_{l}$ attempts of decoding ${\outercw{\tilde{c}}}_{l}$ with $\outerc_{l}$ and a different number of erased symbols in each attempt, which yields a candidate list of at most size $z_{l}$. Finally, the ''best'' candidate $\outercw{\bar{c}}_{l}$ from this list has to be selected using some criterion and inserted into $\mat{\bar{A}}$. 
\end{enumerate}

Note that because of this recurrent structure it is sufficient to consider the BZDA only for ordinary concatenated codes, as decoding of GC codes simply means repeated application of this special case for the sequence $\innerc_{0}, \ldots, \innerc_{l}(\bar{a}_{0, j}, \ldots, \bar{a}_{l-1, j}), \ldots, \innerc_{\ell-1}(\bar{a}_{0, j}, \ldots, \bar{a}_{\ell-2, j})$
of the nested inner subcodes, $j=0, \ldots, \outern-1$, and the corresponding outer codes \cite{dumer:1980, dumer:1998}. The details of the BZDA for ordinary concatenated codes are described in Section \ref{section:BZouterRS}.
In \cite{blokh_zyablov:1982} it is shown that if $\forall\,l\in \{0, \ldots, \ell-1\}:z_{l}=\innerd_{l}/2$ for even $\innerd$, the BZDA can decode up to $\floor{(\outerd_l\innerd_l-1)/2}$ errors in the $l$-th iteration and thus by (\ref{eqn:designeddistance}) also up to $\floor{(\concd-1)/2}$ errors in a GC code word.

\section{Interleaved Reed--Solomon Codes}\label{section:IRS}

Observe that the matrix $\mat{A}$ is just an interleaved set of $\ell$ different RS codes, hence an IRS code. For IRS codes an efficient decoding algorithm was suggested in \cite{schmidt_sidorenko_bossert:2006c}, which has only $\ell$
times the complexity of the Berlekamp--Massey algorithm for decoding one single RS
code. The algorithm allows to correct at most
$r(\ell):=\floor{(\bar{\outerd}-1)\ell/(\ell+1)}$ erroneous columns of $\mat{A}$, where
$\bar{\outerd}:=1/\ell\sum_{l=0}^{\ell-1}\outerd_{l}$ is the average minimum
distance of the interleaved set of RS codes and
\begin{equation}\label{eqn:l1}
\outerd_{l}>r(\ell),\,l=0,...,\ell-1.
\end{equation}
The IRS decoding algorithm from \cite{schmidt_sidorenko_bossert:2006c} yields a decoding failure with some probability. However, this probability can be made small and is neglected here.

If the complete matrix $\mat{A}$ does not satisfy (\ref{eqn:l1}), we can
split it into a number of submatrices with the same length as $\mat{A}$, which all fulfill (\ref{eqn:l1}) and thus can be decoded by the IRS decoding algorithm from \cite{schmidt_sidorenko_bossert:2006c}. Assume that
\begin{equation*}
A_{v}:=\left(
\begin{array}{c}
  \outercw{c}_{v} \\
  \vdots \\
  \outercw{c}_{v+\tilde{\ell}-1}
\end{array}
\right),\;\outercw{c}_l\in \outerc_l
\end{equation*}
is such a submatrix of $\mat{A}$ and forms an IRS code with average minimum distance $\bar{\outerd}$, that satisfies both constraint (\ref{eqn:l1}) and $\bar{\outerd}\innerd_v\tilde{\ell}/(\tilde{\ell}+1)\geq \concd$. The main idea of applying the IRS decoding algorithm from \cite{schmidt_sidorenko_bossert:2006c} to GC codes is as follows: We can replace $\tilde{\ell}$ iterations
$v,...,v+\tilde{\ell}-1$ of the BZDA by the following single iteration:
\begin{enumerate}[(i)]
  \item Decoding of all columns $\innertcw{r}_j$, $j= 0,\ldots, {\outern-1}$, of the received matrix $\mat{R}$ by BMD decoders for the subcodes $\innerc_v(\bar{a}_{0, j}, \ldots, \bar{a}_{v-1, j})$, correcting up to $\floor{(\innerd_v-1)/2}$ errors and yielding an estimate $\mat{\tilde{A}}_v$ for the submatrix $\mat{A}_v$.
 \item Execution of $z_v$ attempts of IRS decoding $\mat{\tilde{A}}_v$ with a different number of erased columns in each attempt, which yields a candidate list of at most size $z_v$. Finally, the ''best'' candidate $\mat{\bar{A}}_v$ from this list has to be selected using some criterion and inserted into $\mat{\bar{A}}$.
\end{enumerate}

As a result of this method, we skipped $\outern(\ell-1)$ inner decodings and we will show that the number of required decoding attempts for the outer code to guarantee decoding up to $\floor{(\outerd\innerd_v-1)/2}$ channel errors is much
smaller than in the original BZDA, which in practice means $z_v\in\{2, 3\}$. Eventually, our modified algorithm corrects up to half the minimum distance of the GC code.

\section{BZDA with Outer BMD Decoding}\label{section:BZouterRS}

In this section, we consider decoding of a {\em simple} concatenated code $\concc$, which consists of the outer RS code $\outerc$ and the inner binary code $\innerc$. This corresponds to the \mbox{$l$-th} iteration of the refined BZDA from \cite{dumer:1980} for a GC code where it holds w.l.o.g $\outerc=\outerc_{l}$, $\innerc=\innerc_{l}(\bar{a}_{0, j}, \ldots, \bar{a}_{l-1, j})$ and $\forall\,i\in\{0, \ldots, l-1\}, j_1, j_2\in\{0, \ldots, \outern-1\}:\bar{a}_{i, j_1}=\bar{a}_{i, j_2}$.

First (step (i) of the BZDA), we decode the columns $\innertcw{r}_j$ of the received matrix $\mat{R}$ by a BMD decoder for $\innerc$, correcting up to $\floor{(\innerd-1)/2}$ errors and yielding code word estimates $\innertcw{\tilde{c}}_j$ or decoding failures. Decoding of the outer code $\outerc$ is performed with respect to the ordered set of thresholds $\{T_1^{(z)}, \ldots, T_z^{(z)}\}$ with $0\leq T_1^{(z)}<\cdots<T_z^{(z)}\leq (\innerd-1)/2$. For each decoding attempt $k\in\{1,\ldots,z\}$
the decoding results of the inner decoder depend on the threshold $T_k$ in the following manner: The symbols $\tilde{r}^\mathrm{o}_j(k)\in \F_{2^m} \cup\{\erasure\}$ delivered to the outer decoder at position $j$ are
\begin{equation}\label{eqn:rtilde}
  \tilde{r}^\mathrm{o}_j(k):=\left\{\begin{array}{ll}
      \mathrm{enc}_{\innerc}^{-1}(\innercw{\tilde{c}}_j), & \Hd{\innercw{r}_j, \innercw{\tilde{c}}_j}\leq T_k,\\
      \erasure,  & \Hd{\innercw{r}_j, \innercw{\tilde{c}}_j}> T_k,\\
      \erasure,  & \text{failure of the inner decoder},
                        \end{array}\right.
\end{equation}
where $\innertcw{r}_j$ is the received word in the $j$-th column, $\innertcw{\tilde{c}}_j$ is the result of inner decoding, $\mathrm{enc}_{\innerc}^{-1}(\cdot)$ maps code words of $\innerc$ to the corresponding $q$-ary information symbols, and $\erasure$ is the symbol for an erasure. As result from outer decoding $\outercw{\tilde{r}}(k):=(\tilde{r}^\mathrm{o}_0(k), \ldots, \tilde{r}^\mathrm{o}_{\outern-1}(k))$ we obtain the outer code word estimate $\outercw{\tilde{c}}(k)=(\tilde{c}^\mathrm{o}_0(k), \ldots, \tilde{c}^\mathrm{o}_{\outern-1}(k))$. From (\ref{eqn:rtilde}) follows that thresholds with equal integers parts yield equal decoding attempts, so the number of $z^\star$ actual attempts may be smaller than the number $z$ of thresholds, i.e. $z^\star\leq z$.

The numbers of decoding errors and erasures occurring at decoding $\outercw{\tilde{r}}(k)$ are denoted by $\outererrors$ and $\outererasures$, respectively. The outer RS code $\outerc$ can successfully decode as long as $2\outererrors+\outererasures<\outerd$, 
since we assume outer BMD decoding in this section. For a fixed number $z$ of thresholds, the following theorem fixes the optimum values of the thresholds such that the decoding bound of the BZDA is maximized.

\begin{theorem}[Blokh, Zyablov \cite{blokh_zyablov:1982}]\label{theorem:thresholdsind}
  For a concatenated code with outer BMD-decoded RS code and inner BMD-decoded code $\innerc(\innern, \innerk, \innerd)$, the set of thresholds $\{T_1^{(z)}, \ldots, T_z^{(z)}\}$ which maximizes the decoding bound is determined by
  \begin{equation}\label{eqn:thresholdsind}
    T_k^{(z)}:=k\cdot \frac{\innerd+1}{2z+1}-1,
  \end{equation}
  $k\in\{1, \ldots, z\}$.
\end{theorem}

If the thresholds are chosen according to (\ref{eqn:thresholdsind}), the decoding bound is given by the following theorem in a sense that the transmitted code word is among the elements of the result list $\mathcal{L}:=\{\outercw{\tilde{c}}(1), \ldots, \outercw{\tilde{c}}(z^\star)\}$.

\begin{theorem}[Blokh, Zyablov \cite{blokh_zyablov:1982}]\label{theorem:boundind}
  For a concatenated code with outer BMD-decoded RS code $\outerc(\outern, \outerk, \outerd)$ and inner BMD-decoded code $\innerc(\innern, \innerk, \innerd)$, the decoding bound is
 \begin{equation}\label{eqn:boundind}
   \concerrors< \outerd(\floor{T^{(z)}_z}+1)=\outerd z\cdot \biggfloor{\frac{\innerd+1}{2z+1}}.
 \end{equation}
\end{theorem}

In Figure \ref{fig:outer_irs_even}, the decoding bound (\ref{eqn:boundind}) is plotted with circles versus the number of thresholds $z$ for an example with $\outerd=33$, $\innerd=20$. It can clearly be seen that the bound reaches $\outerd\innerd/2$, i.e. half the minimum distance of the concatenated code with increasing number of thresholds. The bound obviously only depends on the greatest threshold $T_z^{(z)}$. If we hypothesize that the number of thresholds tends to infinity, we can see that for the greatest threshold
\begin{equation*}
  T_z^{(z)}\underset{z\rightarrow\infty}{\longrightarrow} \frac{\innerd-1}{2}=:T_\infty^{(\infty)}.
\end{equation*}
But as $T_\infty^{(\infty)}=(\innerd-1)/2$, we know that the greatest possible integer threshold $\floor{T_\infty^{(\infty)}}$ is $\innerd/2-1$ if $\innerd$ is even, and $(\innerd-1)/2$ if $\innerd$ is odd. This allows to state the following theorem, which confirms our observation from Figure \ref{fig:outer_irs_even}.

\begin{theorem}\label{theorem:maxboundind}
If the number of thresholds $z$ tends to infinity, the decoding bound of the BZDA for a concatenated code $\concc$ with outer BMD-decoded RS code $\outerc(\outern, \outerk, \outerd)$ and inner BMD-decoded code $\innerc(\innern, \innerk, \innerd)$ is
\begin{equation}\label{eqn:maxboundind}
e < \frac{\outerd\innerd}{2}.
\end{equation}
\end{theorem}
\begin{proof}
The decoding bound (\ref{eqn:boundind}) is non-decreasing in $z$, hence it assumes its maximum at $z\rightarrow\infty$. Consider two cases:
\begin{enumerate}[(i)]
\item $\innerd$ is even, thus the greatest possible integer threshold is $\innerd/2-1$ and $\concerrors < \outerd(\innerd/2-1+1)=\outerd\innerd/2$.
\item $\innerd$ is odd, hence the greatest possible integer threshold is $(\innerd-1)/2$ and $\concerrors < \outerd\big((\innerd-1)/2+1\big)= \outerd\innerd/2+\outerd/2$.
\end{enumerate}
\end{proof}

In the following, we restrict ourselves to binary error matrices $\mat{E}$ meeting (\ref{eqn:maxboundind}).

To obtain decoding bound (\ref{eqn:maxboundind}), the greatest possible integer threshold needs to be among the threshold set. For even $\innerd$ this greatest integer threshold is $T_{\mathrm{even}}:=\floor{T_\infty^{(\infty)}}=\innerd/2-1$, which is strictly smaller than the limit $T_\infty^{(\infty)}$. By the following lemma it can be reached already for a rather small value of $z$.

\begin{lemma}\label{lemma:indeven}
For a concatenated code with inner BMD-decoded code $\innerc(\innern, \innerk, \innerd)$ and outer BMD-decoded RS code $\outerc$ the greatest possible integer threshold $T_{\mathrm{even}}$ is reached if $z\geq \underline{z}:=\innerd/2$.
\end{lemma}
\begin{proof}
Solve $T_z^{(z)}\geq T_{\mathrm{even}}$ for $z$.
\end{proof}

We can thus obtain (\ref{eqn:maxboundind}) with only $\innerd/2$ thresholds according to (\ref{eqn:thresholdsind}) if $\innerd$ is even.

If however $\innerd$ is odd, the greatest possible integer threshold is $T_{\mathrm{odd}}:=T_\infty^{(\infty)}=(\innerd-1)/2$, i.e. the limit $T_\infty^{(\infty)}$ itself. It can obviously only be reached for an infinte number of thresholds.

But the number of integers below $T_{\mathrm{odd}}$ is $(\innerd-1)/2$, hence the number of actual decoding attempts is upper bounded by $(\innerd-1)/2$. It follows that even though in the $\innerd$ odd case the number of required thresholds is infinite, only $z^\star=(\innerd-1)/2$ outer decoding attempts are sufficient to achieve decoding bound (\ref{eqn:maxboundind}).

Up to now, we only know that the transmitted outer code word $\outercw{c}$ is {\em somewhere} within the result list $\mathcal{L}$ of the BZDA if (\ref{eqn:maxboundind}) is fulfilled. The following lemma provides a means of exactly determining its position among the elements of $\mathcal{L}$.

\begin{lemma}[Blokh, Zyablov \cite{blokh_zyablov:1982}]\label{lemma:selection}
Let $t(k):=\sum_{j=0}^{\outern-1}t_j(k)$ with
\begin{equation*}
  t_j(k):=\left\{\begin{array}{ll}
    \Delta_j,            & \text{if}\; \tilde{c}^\mathrm{o}_j(k)=\mathrm{enc}_{\innerc}^{-1}(\innercw{\tilde{c}}_j)\\
    \innerd-\Delta_j,    & \text{if}\; \tilde{c}^\mathrm{o}_j(k)\neq \mathrm{enc}_{\innerc}^{-1}(\innercw{\tilde{c}}_j)\\
    \frac{\innerd}{2}, & \text{failure of the inner decoder},
                \end{array}\right.
\end{equation*}
and $\Delta_j:=\Hd{\innercw{r}_j, \innercw{\tilde{c}}_j}$. Assume $e<\outerd\innerd/2$ and that $T_{k_0}$ is a threshold with $\outercw{\tilde{c}}(k_0)=\outercw{c}$. Then
\begin{equation}\label{eqn:selection}
  t(k_0) < \frac{\outerd\innerd}{2},
\end{equation}
and
\begin{equation*}
  \forall\, k\in\{1,\ldots, |\mathcal{L}|\},\;k\neq k_0:t(k) > \frac{\outerd\innerd}{2}.
\end{equation*}
\end{lemma}

The lemma guarantees that only the transmitted outer code word $\outercw{c}=\outercw{\tilde{c}}(k_0)$ fulfills (\ref{eqn:selection}), i.e. that no further decoding attempts have to be executed as soon as (\ref{eqn:selection}) is fulfilled for the smallest threshold index $k\in\{1, \ldots, z\}$. Then, we set $k_0:=k$ and choose $\outercw{\bar{c}}=\outercw{\tilde{c}}(k_0)$.

\section{BZDA with Outer IRS Codes, i.e. Outer BD Decoding}\label{section:BZouterIRS}

Now we consider the case where $\outerc$ is an IRS code, i.e. a row-wise arrangement of $\ell\geq 2$ RS codes of equal length but potentially different dimensions, which are decoded collaboratively as described in Section \ref{section:IRS}. This allows $\outerc$ to correct a larger number of errors leading to a decoding success whilst $\lambda\outererrors+\outererasures\leq \outerd-1$,
where $1<\lambda:=(\ell+1)/\ell<2$. This means Bounded Distance (BD) decoding. Our aim now is to derive formulae corresponding to (\ref{eqn:thresholdsind}) and (\ref{eqn:boundind}) for this specific case. In doing so, we generalize the approach for outer BMD decoding from \cite{blokh_zyablov:1982}. The procedure is as follows: Let $\concerrorsfail$ be the smallest number of channel errors for a given set of thresholds $\{T_1^{(z)},\ldots, T_z^{(z)}\}$, such that all decoding attempts $k\in\{1,\ldots, z\}$ fail, i.e. such that
\begin{equation}\label{eqn:errorforall}
  \forall\, k\in \{1,\ldots, z\}:\lambda\outererrors+\outererasures>\outerd-1.
\end{equation}
We determine
\begin{equation*}
  \concerrorsfail := \underset{(\outererrorsexplicit{1},\outererasuresexplicit{1},\ldots, \outererrorsexplicit{z}, \outererasuresexplicit{z})}{\min}\{\concerrors\}
\end{equation*}
under the condition that (\ref{eqn:errorforall}) is fulfilled. Then, we find the set $\{T_1^{(z)},\ldots, T_z^{(z)}\}$ of thresholds which maximizes this mini\-mum, i.e. the set of thresholds, which maximizes the decoding bound. This set is determined by the expression
\begin{equation*}
  \{T_1^{(z)}, \ldots, T_z^{(z)}\}:=\underset{\{\bar{T}_1^{(z)}, \ldots, \bar{T}_z^{(z)}\}}{\arg\max}\{\concerrorsfail\}.
\end{equation*}

The detailed derivation is too involved to be presented here, so we confine ourselves to the results in form of the following theorems.

\begin{theorem}\label{theorem:thresholds}
For a concatenated code $\concc$ with outer collaboratively decoded IRS code $\outerc$ consisting of $\ell$ RS codes and inner BMD-decoded code $\innerc(\innern, \innerk, \innerd)$, the set of thresholds $\{T_1^{(z)}, \ldots, T_z^{(z)}\}$ which maximizes the decoding bound is defined by
\begin{equation}\label{eqn:thresholdscol}
  T_k^{(z)} := b-a(\lambda-1)^k
\end{equation}
with
\begin{equation*}
  b := \frac{\innerd-1+\lambda(\lambda-1)^z}{2-\lambda(\lambda-1)^z},\quad
  a := \frac{\innerd+1}{2-\lambda(\lambda-1)^z},
\end{equation*}
$k\in\{1,\ldots, z\}$, where $z$ is the number of thresholds and $1<\lambda=(\ell+1)/\ell<2$.
\end{theorem}

\begin{theorem}\label{theorem:boundcol}
For a concatenated code $\concc$ with outer collaboratively decoded IRS code $\outerc(\outern, \outerk, \outerd)$ consisting of $\ell$ RS codes and $z$ thresholds chosen as in (\ref{eqn:thresholdscol}), the decoding bound is given by
\begin{equation}\label{eqn:boundcol}
  e< \outerd(\floor{T_z^{(z)}}+1).
\end{equation}
\end{theorem}

By Theorem \ref{theorem:boundcol} the decoding bound only depends on threshold $T_z^{(z)}$, the greatest one among the ordered threshold set $\{T_1^{(z)}, \ldots, T_z^{(z)}\}$. Hence to maximize the decoding bound (\ref{eqn:boundcol}) we have to maximize $T_z^{(z)}$. Since the threshold location function (\ref{eqn:thresholdscol}) is non-decreasing, the greatest threshold occurs for $z\rightarrow\infty$, and is $T_\infty^{(\infty)}:=(\innerd-1)/2$. The following theorem states the decoding bound for this greatest possible threshold.

\begin{theorem}\label{theorem:IRShd}
Let $\concc$ be a concatenated code with inner BMD-decoded code $\innerc(\innern, \innerk, \innerd)$ and outer IRS code $\outerc$ with $\ell>2$. If the maximum possible integer threshold is among the threshold set, the decoding bound is given by
\begin{equation*}
\concerrors< \frac{\outerd\innerd}{2}.
\end{equation*}
\end{theorem}
\begin{proof}
Inserting the the integer parts $T_\mathrm{even}:=\floor{T_\infty^{(\infty)}}=\innerd/2-1$ and $T_\mathrm{odd}:=T_\infty^{(\infty)}=(\innerd-1)/2$, respectively, of the greatest possible thresholds into bound (\ref{eqn:boundcol}) proves the statement.
\end{proof}

For even $\innerd$, the greatest possible integer threshold $T_\mathrm{even}$ already is reached considering a finite number of thresholds, i.e. if
$z\geq \underline{z}:=\min\{z\}\;\mathrm{s.t}\;\floor{T_z^{(z)}}=T_\mathrm{even}=\innerd/2-1$. Thus for even $\innerd$ the finite threshold set
\begin{equation}\label{eqn:Tseteven}
\mathcal{T}_\mathrm{even}:=\{T_1^{(\underline{z})}, \ldots, T_{\underline{z}}^{(\underline{z})}\}
\end{equation}
is sufficient to obtain the maximum of (\ref{eqn:boundcol}).

If on the other hand $\innerd$ is odd, the greatest possible integer threshold $T_\mathrm{odd}$ is $T_\infty^{(\infty)}$ itself, hence the number of required thresholds in fact is infinite. But since we know that decoding attempts corresponding to thresholds with equal integer parts coincide, we can skip all thresholds within the interval $(T_\infty^{(\infty)}-1, T_\infty^{(\infty)})$ by the following lemma.

\begin{lemma}\label{lemma:neededthresholds}
For a concatenated code with inner BMD-decoded code $\innerc(\innern, \innerk, \innerd)$ and and outer IRS code $\outerc$ with $\ell$ collaboratively decoded RS codes $T_\infty^{(\infty)}-1=(\innerd-1)/2-1$ is reached if $k\geq \underline{k}:=
\bigceiling{\log_\ell(\innerd+1)}$.
\end{lemma}
\begin{proof}
If $z\rightarrow \infty$, then the threshold location function (\ref{eqn:thresholdscol}) becomes
$T_k^{(\infty)}:=(\innerd-1)/2-(\innerd+1)(\lambda-1)^k/2$. But
$T_k^{(\infty)} \geq (\innerd-1)/2-1\Leftrightarrow k \geq \log_\ell\big((\innerd+1)/2\big)$. 
\end{proof}

By Lemma \ref{lemma:neededthresholds} we know that all thresholds $T_k^{(\infty)}$ in the range $\underline{k}<k<\infty$ have equal integer parts and therefore can be omitted. Thus, instead of the infinite threshold set $\{T_1^{(\infty)}, \ldots, T_\infty^{(\infty)}\}$ it is equivalent to consider the finite set
\begin{equation}\label{eqn:Tsetodd}
\mathcal{T}_\mathrm{odd}:=\{T_1^{(\infty)}, \ldots, T_{\underline{k}}^{(\infty)}\}\cup \{T_\infty^{(\infty)}\}
\end{equation}
with only $\underline{k}+1=\bigceiling{\log_\ell\big((\innerd+1)/2\big)}+1$ elements.

We know that if we utilize the sets $\mathcal{T}_\mathrm{even}$ and $\mathcal{T}_\mathrm{odd}$ of thresholds according to (\ref{eqn:thresholdscol}) for even and odd inner minimum distance $\innerd$, respectively, we can decode up to half the minimum distance of the concatenated code $\concc$. However, the integer parts not of all the thresholds among the sets are necessarily pairwise different. Since decoding attempts in respect to thresholds with equal integer parts coincide, the number $z^\star$ of actual decoding attempts which need to be executed to decode up to half the minimum distance of $\concc$ may be smaller than the number of thresholds. We can calculate it explicitly by
\begin{equation}\label{eqn:actual}
z^\star=\left\{\begin{array}{ll}
  \big|\bigcup_{k=1}^{\underline{z}} \big\{\floor{T_k^{(\underline{z})}}\big\}\big|\leq \underline{z}, & $\innerd$\;\text{even}\\
  \big|\bigcup_{k=1}^{\underline{k}} \big\{\floor{T_k^{(\infty)}}\big\} \cup \big\{\floor{T_\infty^{(\infty)}}\big\} \big|\leq \underline{k}+1, & $\innerd$\;\text{odd}.
\end{array}\right.
\end{equation}

\section{Concluding Examples}\label{section:Example}

Our results are subsumed using the following examples. We assume a concatenated code $\concc$ consisting of an inner code $\innerc(\innern, \innerk, \innerd)$ and an outer code $\outerc(\outern, \outerk, \outerd)$ consisting of $\ell$ rows containing code words of the RS code $\mathcal{RS}(2^8; 255, 223, 33)$.

For {\em even} inner minimum distance $\innerd=20$ the decoding bounds (\ref{eqn:boundind}) for independent outer decoding and (\ref{eqn:boundcol}) for collaborative outer decoding, respectively, depending on the number $z$ of thresholds are shown in Figure \ref{fig:outer_irs_even}. According to Lemma \ref{lemma:indeven}, for independent outer decoding $\underline{z}=10$ thresholds are sufficient to decode up to half the minimum distance of $\concc$. If collaborative decoding of $\ell=2$ outer RS codes is applied, we can calculate the number of required thresholds by $\underline{z}=\min\{z\}\;\mathrm{s.t}\;\floor{T_z^{(z)}}=9$ and get $\underline{z}=3$. Both values are confirmed by the bounds in Figure \ref{fig:outer_irs_even}.

\vskip 0.1cm

\begin{figure}[htbp]
\centering
\includegraphics[width=252pt]{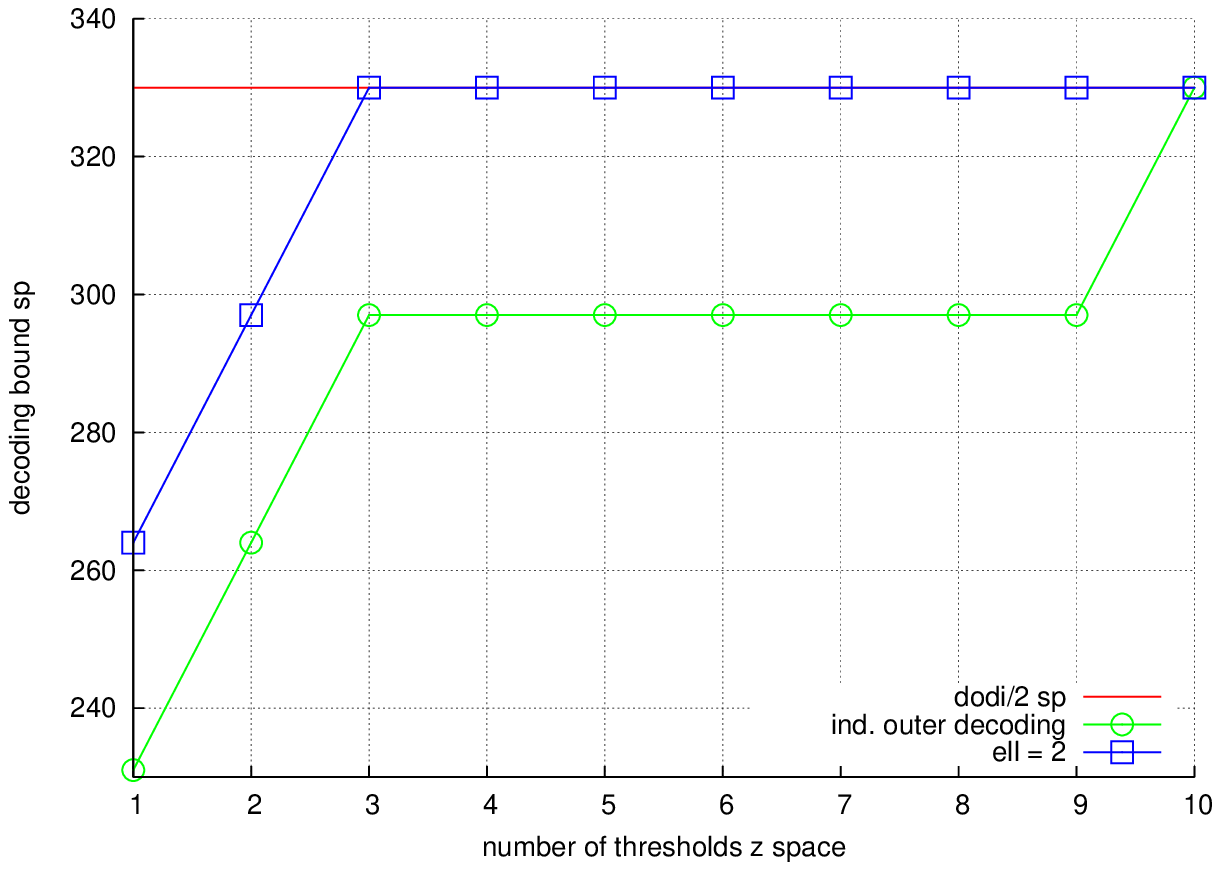}
\caption{Number of thresholds versus decoding bounds (\ref{eqn:boundind}) and (\ref{eqn:boundcol}). The outer code $\outerc$ consists of $\ell=2$ RS codes and $\innerc$ has (even) minimum distance $\innerd=20$.}
\label{fig:outer_irs_even}
\end{figure}

\vskip 0.1cm

If the RS codes are outer codes $\outerc_v, \ldots, \outerc_{v+\ell-1}$ of a GC code as described in Section \ref{section:GCC}, which fulfill (\ref{eqn:l1}), the saving in terms of operations is even greater. Besides the $7$ saved outer decoding attempts, the number of inner decodings can then be cut down by $\outern(\ell-1)=255$.

Note that decoding one IRS code with $\ell$ interleaved RS codes with the algorithm from \cite{schmidt_sidorenko_bossert:2006c} has the same complexity as decoding the $\ell$ RS codes independently. Thus, our comparison of both constructions is fair in terms of complexity.

After establishing the result list $\mathcal{L}$, Lemma \ref{lemma:selection} can be applied to select the transmitted code word among its $|\mathcal{L}|\leq z^\star$ elements.

Figure \ref{fig:attempts} shows the number of actual decoding attempts $z^\star$ as well as the number $\underline{z}$ of thresholds for some reasonable {\em odd} inner minimum distances $\innerd$. Collaborative decoding of $\ell=2$ and $\ell=8$ outer RS codes is considered. For independent outer decoding as described in Section \ref{section:BZouterRS}, $z^\star$ grows linearly with $\innerd$. It diminishes to at most $z^\star=6$ already for an outer IRS code with $\ell=2$. For an outer IRS code with $\ell=8$ already $z^\star=2$ decoding attempts are sufficient to decode up to half the minimum distance of $\concc$ over the full range of all considered odd inner minimum distances $\innerd\in[3, 100]$.

\vskip -0.35cm

\begin{figure}[htbp]
\centering
\includegraphics[width=252pt]{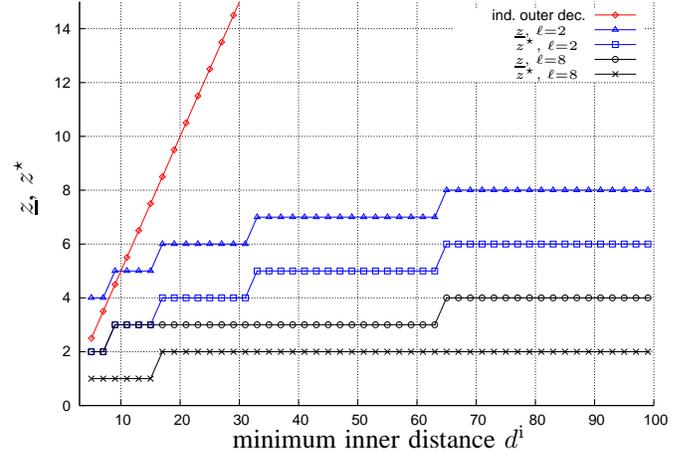}
\caption{Required numbers of thresholds and actual decoding attempts to allow decoding up to half the minimum distance of a concatenated code $\concc$ with parameters as described above. For clarity, only odd $\innerd$ are plotted.}
\label{fig:attempts}
\end{figure}

\bibliographystyle{ieeetr}

\def\noopsort#1{}

\end{document}